\documentclass[doublecolumn]{IEEEtran}

\usepackage{stmaryrd}
\usepackage{amsmath}
\usepackage{amssymb}
\usepackage{graphicx}
\usepackage{todonotes}
\usepackage{cite}
\usepackage{algorithmic}
\usepackage{algorithm}
\usepackage{svg}
\usepackage{mathtools}
\usepackage{caption}
\usepackage{subcaption}

\usepackage{float}
\usepackage{comment}
\usepackage{makecell}
\usepackage{xcolor}
\usepackage{amsthm}
\usepackage{epstopdf}
\usepackage{xcolor}
\usepackage{multirow}

\newtheoremstyle{theoremdd}
  {\topsep}
  {\topsep}
  {\itshape}
  {0pt}
  {\bfseries}
  {.}
  { }
  {\thmname{#1}\thmnumber{ #2}\textnormal{\thmnote{ (#3)}}}
\theoremstyle{theoremdd}

\makeatletter
\renewenvironment{proof}[1][\proofname]{\par
\pushQED{\qed}%
\normalfont \topsep6\p@\@plus6\p@\relax
\trivlist
\item\relax
{\itshape
#1\@addpunct{.}}\hspace\labelsep\ignorespaces
}{%
\popQED\endtrivlist\@endpefalse
}
\makeatother

\usepackage{bm}
\usepackage{bbm}

\newcommand{\argmax}{\mathop{\rm arg~max}\limits}

\newcommand{\bd}{\mathrm{bdiag}}
\newcommand{\tr}{\mathrm{tr}}
\newcommand{\Ber}{\mathrm{Bernoulli}}

\newtheorem{definition}{Definition}
\newtheorem{theorem}{Theorem}
\newtheorem{corollary}{Corollary}

\newcommand{\LRT}[2]{%
  \mathrel{\mathop\gtrless\limits^{#1}_{#2}}%
}

\begin{document}

\title{Activity Detection in Distributed MIMO:\\ Distributed AMP via Likelihood Ratio Fusion}

\author{Jianan~Bai~and~Erik~G.~Larsson
	\thanks{
	The authors are with the Department of Electrical Engineering (ISY), Link\"oping University, 58183 Link\"oping, Sweden (email: jianan.bai@liu.se, erik.g.larsson@liu.se). This work was supported in part by ELLIIT,  the KAW foundation, and  the European Union’s Horizon 2020 research and innovation program under grant agreement no.~101013425 (REINDEER).
	The computations were enabled by resources provided by the Swedish National Infrastructure for Computing (SNIC), partially funded by the Swedish Research Council through grant agreement no.~2018-05973.
}}
	
\maketitle

\begin{abstract}
We develop a new algorithm for activity detection for grant-free multiple access in distributed multiple-input multiple-output (MIMO).
The algorithm is a distributed version of the approximate message passing (AMP) based on a soft combination of likelihood ratios computed independently at multiple access points.
The underpinning theoretical basis of our algorithm is a new observation that we made about the state evolution in the AMP.
Specifically, with a minimum mean-square error denoiser, the state maintains a block-diagonal structure whenever the covariance matrices of the signals have such a structure. 
We show by numerical examples that the algorithm outperforms competing schemes from the literature.
\end{abstract}


\section{Introduction}

\IEEEPARstart{L}{imitless} connectivity is envisioned to be one of the key features in next-generation wireless networks. 
Distributed multiple-input multiple-output (MIMO), also known as  cell-free massive MIMO, is a promising technology to achieve this goal~\cite{cell-free}. 
To support massive connectivity and low latency, grant-free multiple access (GFMA) has been proposed to reduce signaling compared to grant-based access.

In GFMA, an access point (AP) needs to identify the active users and estimate their channels based on received pilots. 
Due to the massive number of devices and the limited coherence block size, assigning mutually orthogonal pilot sequences to all devices becomes impractical. 
The resulting non-orthogonality of the pilots makes the problem of joint activity detection and channel estimation (JADCE)  challenging.
In typical application scenarios, devices only sporadically access the network. This enables JADCE in GFMA  to be cast as a compressed sensing (CS) problem~\cite{liu2018massive}.

For co-located MIMO, GFMA with JADCE has been extensively studied. Specifically, the approximate message passing (AMP) \cite{donoho2009message}, a low-complexity iterative algorithm for CS, has been successfully used \cite{liu2018massive,8961111}. 
An alternative approach is to make a prior assumption on the fading statistics (typically Rayleigh) and find the maximum-likelihood estimates of the signal strengths, using, for example,  coordinate descent~\cite{fengler2021non}. This approach, called ``covariance-based'' in~\cite{fengler2021non} and herein, outperforms AMP for activity detection in co-located MIMO when the number of antennas is large.

For distributed MIMO, several algorithms also exist.  
In particular, the AMP can be applied with an assumption of spatially correlated channels with a block-diagonal covariance that reflects that the APs are geographically distributed \cite{9723198}. 
Extensions of the covariance-based approach to distributed MIMO are also possible, although not uncomplicated. A naive implementation requires solving polynomial equations whose order is proportional to the number of APs. In~\cite{ganesan2021clustering}, the authors proposed a  covariance-based approach with clustering to effectively combine signals from only a subset of APs. 
Eventually, however, all these approaches \cite{9723198,ganesan2021clustering} require all APs to send their received pilot signals to a central processing unit (CPU) and thus incur a high fronthaul load; additionally, they require fully centralized processing of the signals.

We are only aware of one existing paper that attempts to construct a distributed implementation of the AMP for JADCE in GFMA in distributed MIMO: \cite{guo2021joint}. This paper proposed to let each AP run the AMP algorithm to make hard decisions on device activities and then send these decisions to a CPU for fusion. There are also distributed versions of the AMP, developed for more general use cases: \cite{7032167} proposed an iterative solution that requires a two-round communication between the computing nodes and the central node (APs and CPU in our application) and \cite{8462333} that applied a consensus propagation protocol. These  algorithms, however, are inapplicable to GFMA. 

\textbf{Technical contribution:} We develop a distributed version of  AMP and apply it to JADCE for  GFMA in  distributed MIMO. 
By examining the AMP state evolution, we show that when the channels are uncorrelated across different APs, the log-likelihood ratio (LLR) of the device activity is equal to the sum of local LLRs obtained per AP. This enables each AP to compute a local LLR and send it for soft decision fusion. Our algorithm performs closely to the nominal centralized AMP, and it yields the channel estimates as a byproduct.



\section{System Model and Power Allocation}

We consider the uplink of a distributed MIMO system, where $K$ APs jointly serve $N$ single-antenna devices. 
Each AP has $M$ receive antennas, and the total number of antennas in the system is denoted by $M_\textup{tot}=KM$. 
Each device, $n\in\mathcal{N}$, is pre-allocated a pilot sequence $\bm{\phi}_n=[\phi_{1n},\cdots,\phi_{Ln}]^T\in\mathbb{C}^L$ with unit energy, i.e., $\|\bm{\phi}_n \|^2=1$. 
In each time slot, the activity of device $n$ is modeled by a binary random variable, $a_n\sim\Ber(\epsilon_n)$. 


The received signal, $\mathbf{Y}_k\in\mathbb{C}^{L\times M}$, at the $k$-th AP can be expressed as
\begin{equation}
\label{eq: system model}
    \mathbf{Y}_k = \sum_{n\in\mathcal{N}}\sqrt{Lp_n}a_n \bm{\phi}_n\tilde{\mathbf{h}}_{kn}^T + \mathbf{W}_k,
\end{equation}
where $p_n\in[0,p_{\max}]$ is the transmit power of device $n$. The channel between AP $k$ and device $n$ is modeled by $\tilde{\mathbf{h}}_{kn}\sim \mathcal{CN}(\mathbf{0},\tilde{\mathbf{R}}_{kn})$, where
$\tilde{\mathbf{R}}_{kn}\in\mathbb{C}^{M\times M}$ is the spatial correlation matrix, and
$\beta_{kn}=\tr(\tilde{\mathbf{R}}_{kn})/M$ can be interpreted as the large-scale fading coefficient (LSFC). The channel is assumed to be uncorrelated between different APs and devices. The noise matrix $\mathbf{W}_k\in\mathbb{C}^{L\times M}$ has i.i.d. entries with  $\mathcal{CN}(0,\sigma^2)$ elements, where $\sigma^2$ is the noise variance.

For brevity of notation, we define the effective channel ${\mathbf{h}}_{kn} \triangleq \sqrt{Lp_n}\tilde{\mathbf{h}}_{kn}$, which has the distribution $\mathcal{CN}(\mathbf{0},{\mathbf{R}}_{kn})$, where ${\mathbf{R}}_{kn} = Lp_n\tilde{\mathbf{R}}_{kn}$, and $\rho_{kn} = Lp_n\beta_{kn}$ can be interpreted as the received signal strength of device $n$ at AP $k$.

Denoting  the pilot matrix by $\bm{\Phi}=[\bm{\phi}_1,\cdots,\bm{\phi}_N]$,  the effective channel matrix by ${\mathbf{H}}_k=[{\mathbf{h}}_{k1},\cdots,{\mathbf{h}}_{kn}]^T$, and  the vector of device activities by $\mathbf{a}=[a_1,\cdots,a_N]^T$, the received signal model in \eqref{eq: system model} can be written as
\begin{equation}
    \mathbf{Y}_k = \bm{\Phi}\mathbf{D_a}{\mathbf{H}}_k + \mathbf{W}_k.
\end{equation}
By combining the received signal at all APs, we obtain 
\begin{equation}
\label{eq: system model all APs}
    \mathbf{Y} = \bm{\Phi}\mathbf{D_a}\underbrace{[{\mathbf{H}}_1,\cdots,{\mathbf{H}}_K]}_{\triangleq{\mathbf{H}}} + \underbrace{[\mathbf{W}_1,\cdots,\mathbf{W}_K]}_{\triangleq\mathbf{W}},
\end{equation}
where $\mathbf{H}=[\mathbf{h}_1,\cdots,\mathbf{h}_N]^T$ and $\mathbf{h}_n = [\mathbf{h}_{1n}^T,\cdots,\mathbf{h}_{Kn}^T]^T$ is the channel from device $n$ to all APs. Note that by assuming uncorrelated fading across different APs, $\mathbf{h}_n$ has the distribution $\mathcal{CN}(\mathbf{0},\mathbf{R}_n)$, where $\mathbf{R}_n$ is block-diagonal: $\mathbf{R}_n=\bd(\mathbf{R}_{1n},\cdots,\mathbf{R}_{Kn})$.

\subsection{Power Allocation}

In distributed MIMO, since the APs are spread out, the channel gains from a device to different APs vary significantly. 
The signal strength from a device is generally larger at  APs that are physically close to the device than at other APs. 

We propose a user-centric power allocation scheme that comes in a few different variations. The details are as follows:
\begin{itemize}
    \item[1)] Each device $n$ is associated with the subset of APs, say $\mathcal{K}_n^\textup{p}$, for which the LSFCs exceed a threshold $\beta_n^\textup{th}$:
    \begin{equation}
        \mathcal{K}_n^\textup{p} = \{k\in\mathcal{K}: \beta_{kn} > \beta_n^\textup{th} \}.
    \end{equation}
    If no AP satisfies this requirement, we associate the device to the AP with the largest LSFC, i.e.,
    \begin{equation}
    \label{eq: clustering PA}
        \overline{\mathcal{K}}_n^\textup{p} = \mathcal{K}_n^\textup{p} \cup \{\argmax\nolimits_{k\in\mathcal{K}} \beta_{kn} \}.
    \end{equation}
    \item[2)] For each device, a coefficient $s_n$ is calculated. We consider the three different choices:
    \begin{equation}
    \label{eq: power allocation}
        s_n = \left\{ 
        \begin{array}{cl}
            1, & \textit{FullPower} \\
            \max\nolimits_{k\in\overline{\mathcal{K}}_n^\textup{p}} \beta_{kn}, & \textit{MasterAP} \\
            \frac{1}{|\overline{\mathcal{K}}_n^\textup{p} |}\sum\nolimits_{k\in\overline{\mathcal{K}}_n^\textup{p}} \beta_{kn}, & \textit{AvgAP}
        \end{array}
        \right..
    \end{equation}
    \item[3)] For each device, the transmit power is set to 
    \begin{equation}
        p_n = \min\left\{ {s_{\min}}/{s_n}, 1 \right\}  p_{\max},
    \end{equation}
    where $s_{\min} = \min_{n':|\mathcal{K}_n^\textup{p}|\geq 1} s_{n'}$ is the minimum coefficient among all devices for which  at least one AP satisfies the LSFC requirement, i.e., $\beta_{kn}>\beta_n^\textup{th}$.
    
\end{itemize}

\section{Activity Detection in Distributed MIMO}

The system model in \eqref{eq: system model all APs} is an instance of the linear measurement model 
$ \mathbf{Y} = \bm{\Phi}\mathbf{X} + \mathbf{W} $, where the unknown signal matrix $ \mathbf{X}$ is row sparse, and each row $\mathbf{x}_n^T=a_n\mathbf{h}_n^T$ has a Bernoulli-Gaussian distribution.
Therefore, the activity detection becomes a support recovery problem in CS, which can be solved using the AMP algorithm.

\subsection{AMP with MMSE Denoiser and Likelihood-Ratio Test}

By initializing $\mathbf{Z}^{0}=\mathbf{Y}$ and $ \hat{\mathbf{X}}^{0}=\mathbf{O}_{N\times M_\textup{tot}} $, the AMP iteration $ t\in\{0,1,\cdots\} $ for complex-valued signals is \cite{liu2018massive},
\begin{align}\label{key}
	\hat{\mathbf{x}}_n^{t+1} &= \mathbf{g}_t(\underbrace{(\mathbf{Z}^{t})^T\bm{\phi}_n^* + \hat{\mathbf{x}}_n^{t}}_{\triangleq\bm{\xi}_n^t}),\quad\forall n\in\mathcal{N},\\
	\mathbf{Z}^{t+1} &= \mathbf{Y} - \bm{\Phi}\hat{\mathbf{X}}^{t+1} + \frac{1}{L}\mathbf{Z}^{t}\sum_{n\in\mathcal{N}} \mathbf{g}_t'(\bm{\xi}_n^t),
\end{align}
where $ \hat{\mathbf{X}}^t = [\hat{\mathbf{x}}_1^t,\cdots,\hat{\mathbf{x}}_N^t]^T $. Here, $ \mathbf{g}_t(\cdot):\mathbb{C}^{M_\textup{tot}}\rightarrow\mathbb{C}^{M_\textup{tot}} $ is the denoiser and $ \mathbf{g}_t'(\bm{\xi}) $ represents its Jacobian at $ \bm{\xi} $.

As demonstrated in the state evolution analysis~\cite{donoho2009message}, under some mild conditions and in the large-system limit, $\bm{\xi}_n^{t}$ behaves like a Gaussian-noise corrupted version of $\mathbf{x}_n$, i.e.,
\begin{equation}
	\label{eq: distribution of xi}
	\bm{\xi}_n^{t} \sim  \mathbf{x}_n + \mathcal{CN}(\mathbf{0},\bm{\Sigma}^{t}).
\end{equation}
In (\ref{eq: distribution of xi}), $\bm{\Sigma}^{t}$ is referred to as the \textit{state}; this state  evolves by
\begin{equation}
	\label{eq: state evolution}
	\begin{aligned}
		\bm{\Sigma}^{t+1} = \sigma^2\mathbf{I} + \frac{1}{L}\sum_{n\in\mathcal{N}}\mathbb{E}\big[&\big(\mathbf{g}_t(\mathbf{x}_n\! +\! \mathbf{v}^t)-\mathbf{x}_n\big)\\
		&\big(\mathbf{g}_t(\mathbf{x}_n\! +\! \mathbf{v}^t)-\mathbf{x}_n\big)^H \big],
	\end{aligned}
\end{equation}
where $\mathbf{v}^t$ has distribution $\mathcal{CN}(\mathbf{0},\bm{\Sigma}^{t})$ and is independent of $\mathbf{x}_n$, and the expectation is taken over the joint distribution of $\mathbf{x}_n$ and $\mathbf{v}^t$. The initial state is given by
\begin{equation}
	\label{eq: initial state}
	\bm{\Sigma}^{0} = \sigma^2\mathbf{I} + \frac{1}{L}\sum_{n\in\mathcal{N}}\mathbf{R}_n.
\end{equation}

The minimum mean-square error (MMSE) denoiser is given by the MMSE estimate of $\mathbf{x}_n$ given $\bm{\xi}_n^{t}$,
\begin{equation}
	\mathbf{g}_t(\bm{\xi}_n^{t}) = \mathbb{E}[\mathbf{x}_n | \bm{\xi}_n^{t}] = \theta_n^{t}(\bm{\xi}_n^{t}) \cdot \bm{\Psi}_n^{t}\bm{\xi}_n^{t},
	\label{eq: MMSE denoiser}
\end{equation}
where 
{\small
	\begin{equation}
		\theta_n^{t}(\bm{\xi}) = \left({1 + \frac{1\!-\!\epsilon_n}{\epsilon_n}\frac{|\mathbf{R}_n\! +\!  \bm{\Sigma}^{t}|}{|\bm{\Sigma}^{t}|}\exp\big(\!-\!\bm{\xi}^H\bm{\Omega}_n^{t}\bm{\xi} \big)}\!\right)^{-1},
		\label{eq: theta}
	\end{equation}
}
\begin{equation}
	\bm{\Psi}_n^{t} = \mathbf{R}_n(\mathbf{R}_n + \bm{\Sigma}^{t})^{-1},
	\label{eq: Psi}
\end{equation}
\begin{equation}
	\bm{\Omega}_n^{t} = (\bm{\Sigma}^{t})^{-1} - (\mathbf{R}_n + \bm{\Sigma}^{t})^{-1}.
	\label{eq: Omega}
\end{equation}

The support recovery problem is equivalent to the detection of the non-zero entries in the binary vector $\mathbf{a}$.
To determine the value of $a_n$, we consider the binary hypothesis test
\begin{equation}
	\mathcal{H}_0: a_n = 0 \quad\text{and}\quad \mathcal{H}_1: a_n = 1.
\end{equation}
The likelihood-ratio test (LRT) is given by\footnote{For brevity, we henceforth omit the iteration index $t$ in the superscripts. }
\begin{equation}
	\ell_n \triangleq \frac{p(\bm{\xi}_n|a_n=0)}{p(\bm{\xi}_n|a_n=1)} \LRT{\mathcal{H}_0}{\mathcal{H}_1} \gamma,
\end{equation}
where $\gamma>0$ is the decision threshold. According to \eqref{eq: distribution of xi}, the likelihood-ratio can be written as
{\small
	\begin{equation}
		\label{eq: lr}
		\begin{aligned}
			\ell_n = \frac{\mathcal{CN}(\bm{\xi}_n|\mathbf{0},\bm{\Sigma})}{\mathcal{CN}(\bm{\xi}_n|\mathbf{0},\mathbf{R}_n\!+\!\bm{\Sigma})} = \frac{|\mathbf{R}_n\! +\!  \bm{\Sigma}|}{|\bm{\Sigma}|}\exp\big(\!-\!\bm{\xi}_n^H\bm{\Omega}_n\bm{\xi}_n \big).
		\end{aligned}
	\end{equation}
}
Notice that \eqref{eq: theta} can be rewritten as $\theta_n^{t} = (1+\frac{1-\epsilon_n}{\epsilon_n}\ell_n^{t})^{-1}$.

With a large number  of antennas, $M_\textup{tot}$, a naive implementation of the AMP algorithm has two major drawbacks: 1) calculating the determinants and inverting  the $M_\textup{tot}\times M_\textup{tot}$ matrices in \eqref{eq: theta}, \eqref{eq: Psi} and \eqref{eq: Omega} can be computationally demanding; 2) sending the $L\times M_\textup{tot}$-dimensional matrix $\mathbf{Y}$ requires high fronthaul capacity.

\subsection{Covariance Structure in the AMP State Evolution}

The received signal model in distributed MIMO, see  \eqref{eq: system model}, has a special property: the covariance matrices $\{\mathbf{R}_n\}$ are block-diagonal. In the following theorem, we show that during the state evolution in AMP, the states maintain the same block-diagonal structure during all iterations.

\begin{theorem}
\label{th1}
    Assume that $\{\mathbf{R}_n\}$ have a block-diagonal structure: $\mathbf{R}_n=\bd(\mathbf{R}_{1n},\cdots,\mathbf{R}_{Kn})$. By using the MMSE denoiser in \eqref{eq: MMSE denoiser}, the state $\bm{\Sigma}^{t}$ in the state evolution \eqref{eq: state evolution} stays as a block-diagonal matrix with the same structure for each block, i.e., $\bm{\Sigma}^{t} = \bd(\bm{\Sigma}^{t}_1,\cdots,\bm{\Sigma}^{t}_K)$, for all $t$.
\end{theorem}
\begin{proof}
    See Appendix \ref{proof1}.
\end{proof}

According to Theorem \ref{th1}, the inversion of the $M_\textup{tot}\times M_\textup{tot}$ matrices in \eqref{eq: Psi} and \eqref{eq: Omega} can be performed by inverting  their diagonal blocks, which are of  dimension $M\times M$.\footnote{For simplicity, we assume that all APs have the same number of antennas. The algorithm, however, can be easily modified to support arbitrary numbers of antennas.}


When the channel vector from device $n$ to AP $k$ is modeled by i.i.d. Rayleigh fading, the channel covariance matrix becomes $\tilde{\mathbf{R}}_{kn}=\beta_{kn}\mathbf{I}_M$. Correspondingly, the effective channel $\mathbf{h}_n$ from device $n$ to all APs has the distribution $\mathcal{CN}(\mathbf{0},\mathbf{R}_n)$ with $\mathbf{R}_n=\bd(\rho_{1n}\mathbf{I}_M,\cdots,\rho_{Kn}\mathbf{I}_M)$. The following corollary can be viewed as a generalization of~\cite[Theorem 1]{liu2018massive} to the scenario of distributed MIMO.

\begin{corollary}
\label{th2}
    Assume that $\{\mathbf{R}_n\}$ have the diagonal structure $\mathbf{R}_n=\bd(\rho_{1n}\mathbf{I},\cdots,\rho_{Kn}\mathbf{I})$. By using the MMSE denoiser in \eqref{eq: MMSE denoiser}, the state $\bm{\Sigma}^{t}$ stays as a scaled identity matrix for each diagonal block, i.e., $\bm{\Sigma}^{t} = \bd(\tau_1^{t}\mathbf{I},\cdots,\tau_K^{t}\mathbf{I})$, for all $t$.
\end{corollary}
\begin{proof}
    By setting the size of the diagonal blocks in Theorem~\ref{th1} to one, we conclude that the state $\bm{\Sigma}^{t}$ stays as a diagonal matrix. Then, by using the symmetry, we conclude that the elements corresponding to the same AP are equal.
\end{proof}
In the i.i.d. Rayleigh case, the calculations of all matrix inversions and determinants simplify to scalar operations.

\subsection{Distributed Activity Detection}

Since by Theorem \ref{th1}, $ \bm{\Sigma} $ and $\bm{\Omega}_n$ are both block-diagonal, we can rewrite the likelihood-ratio in \eqref{eq: lr} as
\begin{equation}
\begin{aligned}
\label{eq: factorized lr}
    \ell_n 
    &= \prod_{k\in\mathcal{K}}\underbrace{\frac{|\mathbf{R}_{kn}\! +\!  \bm{\Sigma}_k|}{|\bm{\Sigma}_k|}\exp\big(\!-\!\bm{\xi}_{kn}^H\bm{\Omega}_{kn}\bm{\xi}_{kn} \big)}_{\triangleq \ell_{kn}}.
\end{aligned}
\end{equation}
Equivalently, the LLR is
\begin{equation}
\label{eq: factorized llr}
    \log \ell_n = \sum_{k\in\mathcal{K}} \log \ell_{kn},
\end{equation}
where
\begin{equation}
    \log\ell_{kn} = \log \frac{|\mathbf{R}_{kn}\! +\!  \bm{\Sigma}_k|}{|\bm{\Sigma}_k|} - \bm{\xi}_{kn}^H\bm{\Omega}_{kn}\bm{\xi}_{kn}.
\end{equation}
Here, $\bm{\Sigma}_k$ and $\bm{\Omega}_{kn}$ are the $k$-th diagonal blocks of $\bm{\Sigma}$ and $\bm{\Omega}_n$, respectively, and $\bm{\xi}_{kn}$ is the corresponding subvector of $\bm{\xi}_n$.
This means that the LLR $\log\ell_n$ can be written as the sum of $\{\log\ell_{kn} \}$, which can be interpreted as the \textit{local} LLRs after \textit{coherently} processing the received signals at each AP.

In the special case of i.i.d. Rayleigh fading, the LLR can be further simplified into
\begin{equation}
    \log\ell_{kn} = M\log\left(1\!+\!\frac{\rho_{kn}}{\tau_k}\right) - \frac{\rho_{kn}\|\bm{\xi}_{kn}\|^2}{\tau_k(\rho_{kn}\!+\!\tau_k)},
\end{equation}
where the quantity $\rho_{kn}/\tau_k$ can be interpreted as the signal-to-noise ratio (SNR).

Inspired by the factorization in \eqref{eq: factorized lr}, we propose a distributed approach to activity detection in distributed MIMO. The procedure is as follows: each AP runs the AMP algorithm locally by using only the received signal $\mathbf{Y}_k$ and sends the local LLR $\log\hat{\ell}_{kn}$ to the aggregator.
Then, the aggregator computes the LLR $\log\hat{\ell}_n = \sum_{k\in\mathcal{K}}\log\hat{\ell}_{kn}$ for activity detection.

\begin{algorithm}[t]
	\caption{distributed AMP (dAMP)}
	\begin{algorithmic}[1]
		\label{alg: dAMP}
		\REQUIRE $\bm{\Phi}$,$\{\mathbf{Y}_k\}$,$\{\mathbf{R}_{kn}\}$
		\ENSURE {\small$\mathbf{Z}_k^{0}=\mathbf{Y}_k$, $\hat{\mathbf{x}}_{kn}^t=\mathbf{0}$, and $\bm{\Sigma}_k^{0}=\frac{1}{L}\mathbf{Y}_k^T\mathbf{Y}_k^*$, $\forall k,\forall n$}
		\FOR{each $k\in\mathcal{K}$, independently}
		\FOR{$t=0,1,\cdots$}
		\FOR{each $ n\in\mathcal{N}_k^\textup{d} $}
		\STATE $\bm{\xi}_{kn}^{t} = (\mathbf{Z}_k^{t})^T\bm{\phi}_n^* + \mathbf{x}_{kn}^{t}$ \label{step1}
		\STATE $\bm{\Psi}_{kn}^{t} = \mathbf{R}_{kn}(\mathbf{R}_{kn} + \bm{\Sigma}_k^{t})^{-1}$ \label{line1}
		\STATE $\bm{\Omega}_{kn}^{t} = (\bm{\Sigma}^{t}_k)^{-1} - (\mathbf{R}_{kn} + \bm{\Sigma}_k^{t})^{-1}$ 
		\STATE $\ell_{kn}^{t} = \frac{|\mathbf{R}_{kn}\! +\!  \bm{\Sigma}_k^{t}|}{|\bm{\Sigma}_k^{t}|}\exp\big(\!-\!(\bm{\xi}_{kn}^{t})^H\bm{\Omega}_{kn}^{t}\bm{\xi}_{kn}^{t} \big)$ \label{line2}
		\STATE $\theta_{kn}^{t} = \big(1+\frac{1-\epsilon_n}{\epsilon_n}{\ell_{kn}^{t}} \big)^{-1}$
		\STATE $\hat{\mathbf{x}}_{kn}^{t} = \theta_{kn}^{t}\bm{\Psi}_{kn}^{t}\bm{\xi}_{kn}^{t}$ \label{step2}
		\ENDFOR
		\STATE {\small$\mathbf{U}_k^{t} = \frac{1}{N}\sum_{n\in\mathcal{N}_k^\textup{d}}\theta_{kn}^{t}\bm{\Psi}_{kn}^{t}(\mathbf{I}\!+\!(1\!-\!\theta_{kn}^{t})\bm{\xi}_{kn}^{t}(\bm{\xi}_{kn}^{t})^H\bm{\Omega}_{kn}^{t})$}
		\STATE $\mathbf{Z}_k^{t+1} =  \mathbf{Y}_k - \sum_{n\in\mathcal{N}_k^\textup{d}}\bm{\phi}_n(\mathbf{x}_{kn}^{t})^T + \frac{N}{L}\mathbf{Z}_k^{t}\mathbf{U}_k^{t}$ \label{step3}
		\STATE $\bm{\Sigma}_k^{t+1} = \frac{1}{L}(\mathbf{Z}_k^{t+1})^T(\mathbf{Z}_k^{t+1})^*$ \label{step4}
		\ENDFOR
		\ENDFOR
	\end{algorithmic}
\end{algorithm}


\begin{figure*}[h]
	\centering
	\begin{subfigure}[b]{0.49\textwidth}
		\centering
		\includegraphics[width=\textwidth]{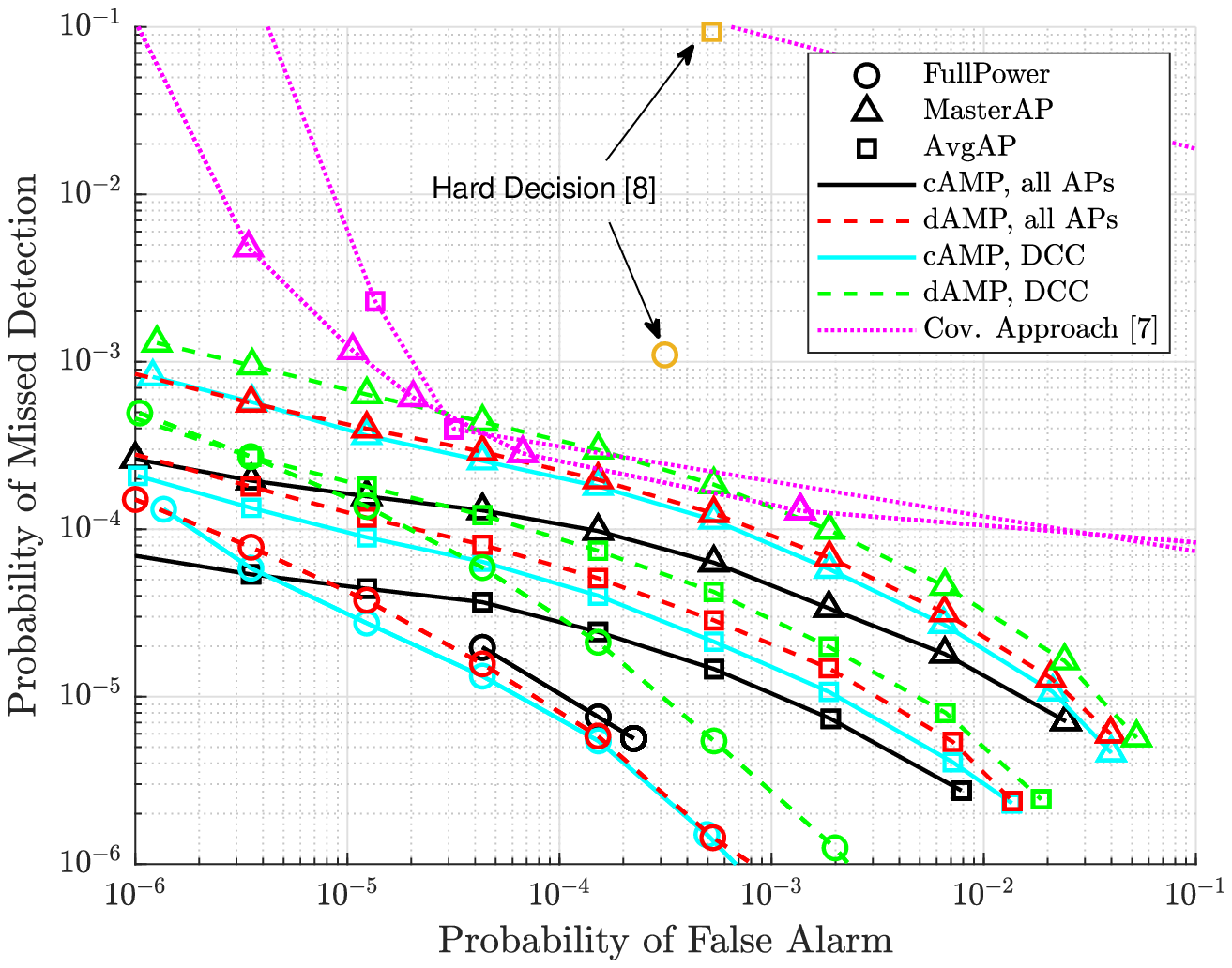}
		\caption{Pilot length $L=40$}
		\label{fig: result L=40}
	\end{subfigure}
	\hfill
	\begin{subfigure}[b]{0.49\textwidth}
		\centering
		\includegraphics[width=\textwidth]{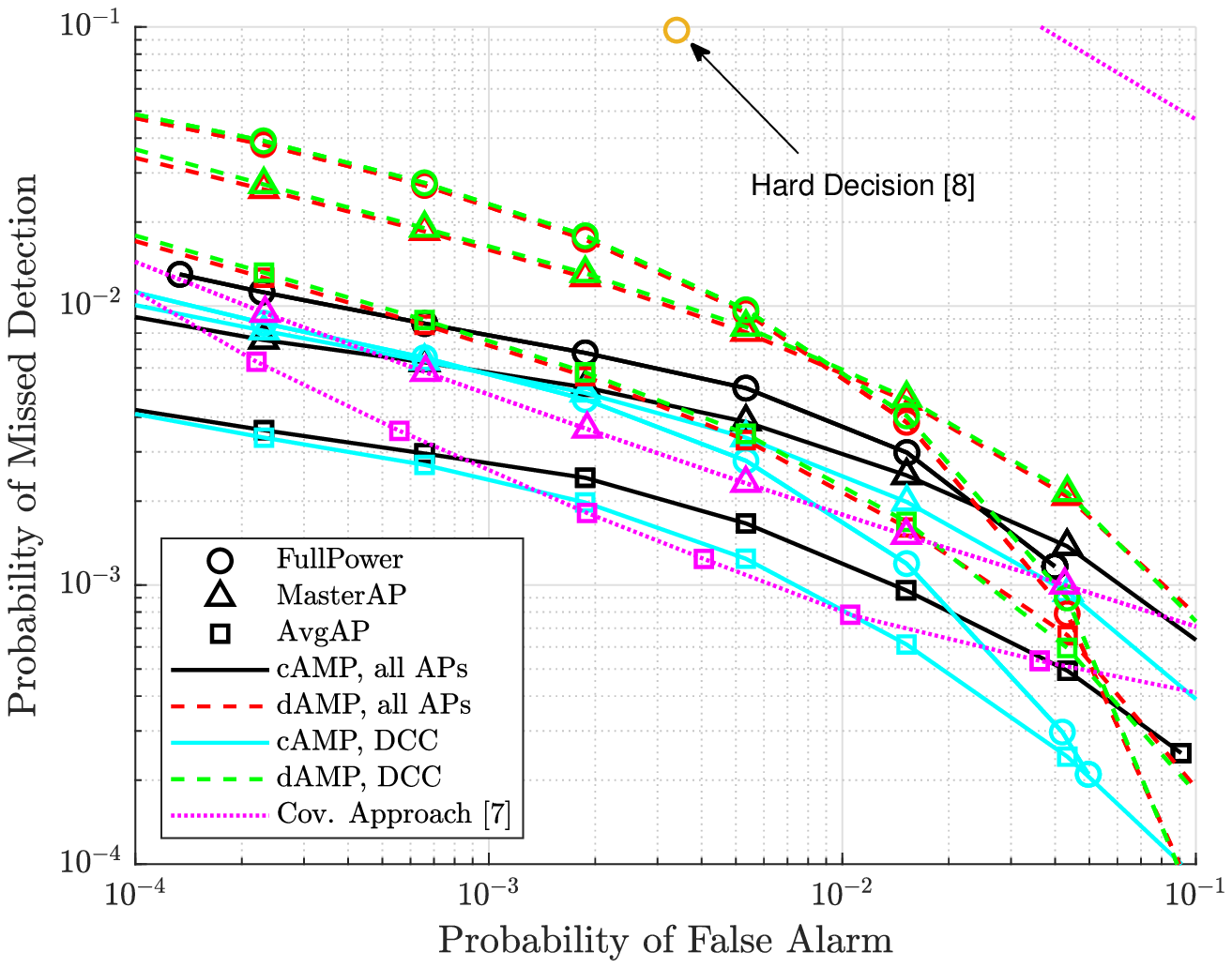}
		\caption{Pilot length $L=20$}
		\label{fig: result L=20}
	\end{subfigure}
	\caption{Performance of the cAMP, dAMP, and baseline algorithms  with different power allocation schemes.}
	\label{fig: result}
\end{figure*}

\subsection{Dynamic Cooperation Clustering}

We assumed that each device was served by all APs. This configuration is not scalable in complexity and resource requirements as $N\rightarrow\infty$. 
Meanwhile, the AMP algorithm, or more generally, CS techniques, are known to work in the regime where the measurement size (pilot length) is larger than or equal to the support size (number of active devices).

To address these issues, we consider a dynamic cooperation clustering (DCC) framework, such that a device is served only by the APs with indices in the set $\mathcal{K}_n^\textup{d}\subset\mathcal{K}$.
Conversely, an AP only serves a subset of devices $\mathcal{N}_k^\textup{d}=\{n\in\mathcal{N}:k\in \mathcal{K}_n^\textup{d}\}$. There are two advantages of using the DCC framework: 1) the computational complexity is reduced; 2) the effective number of active devices served by an AP decreases.

Finally, by exploiting the DCC framework and our new findings about the MMSE denoiser, we propose a distributed AMP (dAMP) which is detailed in Algorithm \ref{alg: dAMP}. A centralized AMP (cAMP) is also developed in a similar way, while the step-wise details are omitted owing to  space constraints. The key distinction in cAMP is that the denoiser for device $ n $ is designed using
$
	\theta_{n}^{t} = \big(1+\frac{1-\epsilon_n}{\epsilon_n}{\prod_{k\in\mathcal{K}_n^\textup{d}}\ell_{kn}^{t}} \big)^{-1}
$
by combining the local LLRs from all its serving APs in each iteration.
These algorithms can be modified for other network structures. For example, multiple neighboring APs can coherently process the received signals. In this respect, cAMP (fully coherent) and dAMP (noncoherent) represent two extreme cases.

\subsection{Complexity Analysis}

The computational complexity of dAMP with correlated fading  is dominated by the calculation of matrix inversions and determinants in steps \ref{line1}-\ref{line2} of Algorithm \ref{alg: dAMP} with complexity $ O(M^3) $. 
Therefore, the overall complexity is $ O(KTNM^{3}) $. For the i.i.d. Rayleigh case, the complexity of matrix-vector multiplications in steps \ref{step1}, \ref{line2}, and \ref{step2} is $ O(M^2) $, and the matrix multiplications in steps \ref{step3} and \ref{step4} have complexity $ O(LM^2) $. Since we are interested in the regime where $ L\ll N $, the overall complexity becomes $ O(KTNM^2) $. Notice that dAMP can be distributed, and the processing per  AP has complexity $ O(TNM^2) $. Furthermore, by using the DCC framework, we can replace $ N $ by $ \max_{k}|\mathcal{N}_k^\textup{d}| $. 

For comparison, the covariance-based method in~\cite{ganesan2021clustering} has overall complexity $ O(TN(K_\textup{dom}^3 + KL^2)) $, where $ K_\textup{dom} $ is the number of dominants APs; for the typical case $ K_\textup{dom}<L $, the complexity becomes $ O(KTNL^2) $. Note, however, that method of \cite{ganesan2021clustering}  is developed  for the i.i.d. Rayleigh case and while extensions are possible, they are likely to incur higher complexity.  
Since the number of antennas is typically small on an AP, we have $ M<L $, and our algorithms have lower complexity than that of \cite{ganesan2021clustering}.

\begin{table}[t]
	\centering
	\begin{tabular}{|c|c|c|c|c|c|}
		\hline
		\multirow{2}{*} & \multicolumn{2}{c|}{\textbf{cAMP}} & \multicolumn{2}{c|}{\textbf{dAMP}} & \multirow{2}{*}{\textbf{Cov. approach}}  \\
		\cline{2-5}
		& all APs & DCC & all APs & DCC & \\
		\hline
		$L=40$ & $0.69$ & $0.28$ & $0.38$ & $0.21$ & $2.33$ \\
		\hline
		$L=20$ & $0.70$ & $0.29$ & $0.37$ & $0.21$ & $1.34$ \\
		\hline
	\end{tabular}
	\caption{Runtime comparison in seconds.}
	\label{tab: runtime}
\end{table}

\section{Simulations}

We consider a distributed MIMO system with $K=20$ APs with $M=3$ antennas each. A total of $N=400$ devices are randomly dropped in a $2~\text{km}\times2~\text{km}$  squared area with activity probability $\epsilon_n=0.1,\forall n $. By using a wrap-around
technique, we approximate an infinitely large network with $15$ antennas and $10$ active devices per square km. The pilots are random Gaussian sequences normalized to unit energy. The maximum transmit power is $23$~dBm. The bandwidth is $1$~MHz. The noise power spectral density is $-169$~dBm/Hz. The LSFC is generated by $-140.6 - 36.7\log_{10}(d_n) +  \Upsilon_i$ in dB, where $d_n$ is the distance from device $n$ to the AP in km, and $\Upsilon_n$ is the shadow fading effect distributed as $\mathcal{N}(0,\sigma_\textup{sf}^2)$, with standard deviation $\sigma_\textup{sf}=4$~dB. The small-scale fading is modeled by i.i.d. Rayleigh for each pair of AP and device.
The LSFC threshold for power allocation is set to satisfy $p_{\max}\beta_n^\textup{th}=6$~dB, $\forall n$. For the DCC framework, we connect each device to the $10$ APs with the largest LSFC.

The performances of cAMP and dAMP are examined in Fig.~\ref{fig: result} with or without the DCC framework and with different power allocation schemes.\footnote{Code available at https://github.com/jiananbai/distributed-AMP.} The covariance-based approach in~\cite{ganesan2021clustering} (with $3$ dominant APs) and the hard-decision-and-fusion based AMP method\footnote{Since \cite{guo2021joint} provided neither theoretical results nor algorithm details for the multi-antenna AP case, we use the expressions of probabilities of missed detection and false alarm in~\cite{liu2018massive} to perform the decision fusion. Notice that this method uses a minimum-probability-of-error criterion and does not produce receiver operating characteristic (ROC) curves.} in~\cite{guo2021joint} are used as baselines for comparison. A runtime comparison is provided in Table~\ref{tab: runtime}.\footnote{The simulations were performed on an Intel Xeon Gold 6130 Processor.}

The results for pilot length $L=40$ are shown in Fig. \ref{fig: result L=40}. The following observations can be made: (i) When the pilot length is larger than or equal to the average number of active devices, AMP outperforms the covariance-based approach in almost all configurations since our AMP algorithms can  coherently process received signals from more APs. (ii) AMP works better with full power. We hypothesize that this is because of the macro-diversity in distributed MIMO: for each device there are almost always some APs to which the  path gain is high. 
Thus, although using full power usually works poorly for activity detection in co-located MIMO, it can be an option in distributed MIMO where it is difficult to obtain an explicit objective for optimizing the power allocation. (iii) There is a performance gap between cAMP and dAMP due to the lack of coherent processing across different APs for dAMP.

In Fig. \ref{fig: result L=20}, the results are reproduced for pilot length $L=20$, which is not a working regime for AMP in the co-located case. 
We observe: 1) the performance loss of AMP is more significant than for the covariance-based approach since the AMP is inherently restricted to scenarios where $ L\geq\sum a_n $, while the covariance-based approach has a better scaling law~\cite{fengler2021non}; 
2) \textit{AvgAP} becomes a better power allocation scheme, potentially due to its better control of interference power under  increased pilot contamination; 
3) DCC performs better than using all APs, since an AP can ignore the devices with bad channel conditions; this is particularly helpful when 
$L$ is small relative to the average number of active users.

\section{Conclusion}

We showed that for activity detection in distributed MIMO: 1) the AMP algorithm can be implemented in a distributed manner with an acceptable performance loss; 2) the AMP algorithm outperforms the covariance-based approach when the pilot length is larger than or equal to the average number of active devices, although this is not the case in co-located MIMO; 3) the problem of pilot correlation can be alleviated by using the DCC framework, when the pilot length is less than the average number of active devices.

\appendices
\section{}
\label{proof1}

We prove Theorem \ref{th1} by induction. First, when the covariance matrices $\{\mathbf{R}_n\}$ share a block-diagonal structure, the initial state $\bm{\Sigma}^{0}$ in \eqref{eq: initial state} has the same block-diagonal structure.
Then, assuming that $\bm{\Sigma}^{t}$ stays in this structure, we show that $\bm{\Sigma}^{t+1}$ has the same structure.

\begin{definition}
    {(Partially Odd or Even Function)} A function $f:\mathbb{R}^M\rightarrow\mathbb{R}$ is partially odd or even in indices $\mathcal{I}\subset\mathcal{M}=\{1,\cdots,M \}$ if $f(\eta_{\mathcal{I}}(\mathbf{x}))=-f(\mathbf{x})$ or $f(\eta_{\mathcal{I}}(\mathbf{x}))=f(\mathbf{x})$, respectively. Here, $\eta_\mathcal{I}(\cdot)$ is an element-wise operator with $[\eta_\mathcal{I}(\mathbf{x})]_i$ equals to $-x_i$ for $i\in\mathcal{I}$, and $x_i$ otherwise.
\end{definition}

An arbitrary expectation term in the summand of the second term in the state evolution \eqref{eq: state evolution} can be written as
\begin{equation}
\label{eq: four terms}
\begin{aligned}
    & \mathbb{E}\Big[\big(\mathbf{g}(\mathbf{x}\! +\! \mathbf{v},\bm{\Sigma})-\mathbf{x}\big)\big(\mathbf{g}(\mathbf{x}\! +\! \mathbf{v},\bm{\Sigma})-\mathbf{x}\big)^H \Big]\\
    &= \mathbb{E}\Big[\mathbf{g}(\mathbf{x}\! +\! \mathbf{v},\bm{\Sigma})\big(\mathbf{g}(\mathbf{x}\! +\! \mathbf{v},\bm{\Sigma})\big)^H \Big] + \mathbb{E}\big[\mathbf{x}\mathbf{x}^H \big]\\
    & ~~~ -\mathbb{E}\big[\mathbf{g}(\mathbf{x}\! +\! \mathbf{v},\bm{\Sigma})\mathbf{x}^H\big] - \mathbb{E}\Big[\mathbf{x}\big(\mathbf{g}(\mathbf{x}\! +\! \mathbf{v},\bm{\Sigma})\big)^H\Big].
\end{aligned}
\end{equation}
According to \eqref{eq: MMSE denoiser}, the first term equals
\begin{equation}
    \bm{\Psi}\underbrace{\mathbb{E}\big[\theta(\mathbf{x}+\mathbf{v})^2(\mathbf{x}+\mathbf{v})(\mathbf{x}+\mathbf{v})^H \big]}_{\triangleq \mathbf{Q}}\bm{\Psi}^H,
\end{equation}
where $\theta(\cdot)$ is defined in \eqref{eq: theta}. By denoting as $p_x(\mathbf{x})$ and $p_v(\mathbf{v})$ the density functions of $\mathbf{x}$ and $\mathbf{v}$, respectively, the $(i,j)$-th element of $\mathbf{Q}$ is given by
\begin{equation}
    [\mathbf{Q}]_{i,j} = \int_{\mathbf{x},\mathbf{v}} \underbrace{(x_i\!+\!v_i)(x_j\!+\!v_j)^*\theta(\mathbf{x}\!+\!\mathbf{v})^2p_x(\mathbf{x})p_v(\mathbf{v})}_{\triangleq f_{i,j}(\mathbf{x},\mathbf{y})}.
\end{equation}
Denote by $\mathcal{M}_k$ the row (column) indices corresponding to the $k$-th diagonal block. Then $f_{i,j}(\mathbf{x},\mathbf{y})$ is partially odd in $\mathcal{M}_k$ if $i\in\mathcal{M}_k$ and $j\notin\mathcal{M}_k$, or $i\notin\mathcal{M}_k$ and $j\in\mathcal{M}_k$, and partially even in $\mathcal{M}_k$ otherwise. That is, $[\mathbf{Q}]_{i,j}=0$ if the indices $i$ and $j$ are not in the same diagonal block. This means that $\mathbf{Q}$ is also block-diagonal with the same structure as $\{\mathbf{R}_n \}$. Then, the first term in \eqref{eq: four terms}, which equals to $\bm{\Psi}\mathbf{Q}\bm{\Psi}^H$, keeps the same block-diagonal structure.

By using similar arguments, one can show that the remaining terms have the same structure. 
We omit the details due to the limited space.

\bibliographystyle{IEEEtran}
\bibliography{Bai_WCL2022-1199_ref}

\end{document}